\documentclass[peerreview,12pt,draftclsnofoot,onecolumn]{IEEEtran}

\usepackage{setspace}

\newtheorem{theorem}{\bf Theorem}[section]
\newtheorem{lemma}[theorem]{\bf Lemma}

\newtheorem{conjecture}[theorem]{\bf Conjecture}

\newtheorem{definition}{\bf Definition}[section]

\newcommand{\beq}{\begin{equation}}
\newcommand{\eeq}{\end{equation}}
\newcommand{\beqq}{\[}
\newcommand{\eeqq}{\]}
\newcommand{\barr}{\begin{array}}
\newcommand{\earr}{\end{array}}
\newcommand{\berr}{\begin{eqnarray}}
\newcommand{\eerr}{\end{eqnarray}}
\newcommand{\berrr}{\begin{eqnarray*}}
\newcommand{\eerrr}{\end{eqnarray*}}

\begin{document}

\begin{spacing}{1.494}

\title{Bucketing Coding and Information Theory for the Statistical High
Dimensional Nearest Neighbor Problem}

\author{Moshe~Dubiner% <-this % stops a space
\thanks{M. Dubiner is with Google, e-mail: moshe@google.com}% <-this % stops a space
\thanks{Manuscript submitted to IEEE Transactions on Information Theory on March 3 ,2007; revised August 27, 2007.}
}

\maketitle

\section*{Acknowledgment}

This paper is dedicated to my wife Edith whose support made it possible,
and to Benjamin Weiss who taught me what mathematical information means.
I also thank Uri Zwick for suggesting clarifications.

% \fi
\IEEEpeerreviewmaketitle

\begin{abstract}

%\boldmath

Consider the problem of finding high dimensional approximate nearest
neighbors, where the data is generated by some known probabilistic model.
We will investigate a large natural class of algorithms which we call
bucketing codes.
We will define bucketing information, prove that it bounds
the performance of all bucketing codes, and that the bucketing information
bound can be asymptotically attained by randomly constructed bucketing codes.

For example suppose we have $n$ Bernoulli(1/2) very long
(length $d\rightarrow\infty$) sequences of bits. Let $n-2m$ sequences be
completely independent, while the remaining $2m$ sequences are composed of $m$
independent pairs. The interdependence within each pair is that
their bits agree with probability $1/2<p\le 1$. It is well known how to find
most pairs with high probability by performing order of
$n^{\log_{2}2/p}$ comparisons. We will see that order of
$n^{1/p+\epsilon}$ comparisons suffice, for any $\epsilon>0$.
Moreover if one sequence out of each pair belongs to a a known set of
$n^{(2p-1)^{2}-\epsilon}$ sequences, than pairing can be done
using order $n$ comparisons!

\end{abstract}

%\begin{IEEEkeywords}
%IEEEtran, journal, \LaTeX, paper, template.
%\end{IEEEkeywords}

\section{Introduction}  \label{introduction}

Suppose we have two bags of points, $X_{0}$ and $X_{1}$,
randomly distributed in a high-dimensional space. The points are independent of
each other, with one exception: there is one unknown point $x_{0}$
in bag $X_{0}$ that is significantly closer to an unknown point $x_{1}$
in bag $X_{1}$ than would be accounted for by chance. We want an efficient
algorithm for quickly finding these two 'paired' points.
More generally, one could have $m$ special pairs (up to having all points
paired). An algorithm that finds a single pair with probability $S$ will find
an expected number of $mS$ pairs, so keeping $m$ as a parameter is unnecessary.

We worked on finding texts that are translations of each other,
which is a two bags problem (the bags are languages).
In most cases there is only one bag $X_{0}=X_{1}=X$, $n_{0}=n_{1}=n $.
The two bags model is slightly more complicated,
but leads to clearer thinking. It is a bit reminiscent
of fast matrix multiplication: even when one is interested only in
square matrices, it pays to consider rectangular matrices too.

Let us start with the well known simple uniform marginally Bernoulli(1/2)
example. Suppose $X_{0},X_{1}\subset \{0,1\}^{d}$ of sizes $n_{0},n_{1}$
respectively are randomly chosen as independent Bernoulli(1/2) variables,
with one exception. Choose uniformly randomly one point $x_{0}\in X_{0}$,
xor it with a random Bernoulli($p$) vector and overwrite one uniformly chosen
random point $x_{1}\in X_{1}$. A symmetric description is to say that
$x_{0},x_{1}$ $i$'th bits have the joint probability matrix
\beq P=\left( \barr{cc}
p/2 & (1-p)/2 \\
(1-p)/2 & p/2
\earr \right) \label{ber}\eeq
for some known $1/2<p\le 1$.
In practice $p$ will have to be estimated.

Let
\beq \ln N=\ln n_{0}+\ln n_{1}-I(P)d \label{info}\eeq
where
\beq I(P)=I(p)=p\ln(2p)+(1-p)\ln(2(1-p)) \label{muty}\eeq
is the mutual information between the special pair's single coordinate values.
Information theory tells us that we can not hope to pin the special pair down
into less than $N$ possibilities, but can come close to it in some asymptotic
sense. Assume that $N$ is small. How can we find the closest pair?
The trivial way to do it is to compare all the $n_{0}n_{1}$ pairs.
A better way has been known for a long time. The earliest references I am
aware of are Karp,Waarts and Zweig \cite{KWZ95}, Broder \cite{Bro98},
Indyk and Motwani \cite{IM98}. They do not
limit themselves to this simplistic problem, but their approach clearly
handles it. Without restricting generality let $n_{0}\le n_{1}$.
Randomly choose
\beq k \approx \log_{2}n_{0} \eeq
out of the $d$ coordinates, and compare the point pairs which agree on
these coordinates (in other words, fall into the same bucket).
The expected number of comparisons is
\beq n_{0}n_{1}2^{-k}\approx n_{1} \eeq
while the probability of success of one comparison is $p^{k}$.
In case of failure try again, with other random k coordinates. At first
glance it might seem that the expected number of tries until success
is $p^{-k}$, but that is not true because the attempts are interdependent.
An extreme example is $d=k$, where the attempts are identical.
In the unlimited data case $d\rightarrow\infty$ the expected number of tries
is indeed $p^{-k}$, so the expected number of comparisons is
\beq W\approx p^{-k}n_{1}\approx n_{0}^{\log_{2}1/p}n_{1} \eeq
Is this optimal? Alon \cite{NA} has suggested the possibility of improvement
by using Hamming's perfect code.

We have found that in the $n_{0}=n_{1}=n$
case, $W\approx n^{\log_{2}2/p}$ can be reduced to
\beq W\approx n^{1/p+\epsilon} \eeq
for any $1/2<p<1$, $\epsilon>0$. This particular algorithm is described in
the next section. Amazingly it is possible to characterize the
asymptotically best exponent not only for this problem,
but for a much larger class. We allow non binary discrete data, a limited
amount of data ($d<\infty$) and a general probability distribution of each
coordinate.

We will prove theorem \ref{big}, a lower bound on the work performed by any
bucketing algorithm. It employs a newly defined $\mathbf{bucketing}$
$\mathbf{information}$ function $I(P,\lambda_{0},\lambda_{1},\mu)$, which
generalizes Shannon's mutual information function $I(P)=I(P,1,1,\infty)$.
Comparing (\ref{info}) with theorem \ref{big} shows that the mutual
information's function generalizes as well. Bucketing algorithms approaching
the information bound are constructed by random coding. The analogy with
Shannon's coding and information theory is very strong, suggesting that maybe
we are redoing it in disguise. If it is a disguise, it is quite effective.
Coding with distortion theory seems also related.
There is related work \cite{MD1}, which tackles a particular class of
practical bucketing algorithms (lexicographic forest algorithms).
Their performance turns out to be bounded by a $\mathbf{bucketing}$
$\mathbf{forest}$ $\mathbf{information}$ function, and that bound is
asymptotically attained by a specific practical algorithm.

\section{An Asymptotically Better Algorithm} \label{better}

The following algorithm does not generalize well, but makes sense
for the uniform marginally Bernoulli(1/2) problem (\ref{ber}) with $1/2<p<1$.
Let $0<d_{0}\le d$ be some natural numbers.
We construct a $d$ dimensional bucket in the following way.
Choose a random point $b\in\{0,1\}^{d}$. The bucket contains all points
$x\in\{0,1\}^{d}$ such for exactly $d_{0}-1$ or $d_{0}$  coordinates $i$
$x_{i}=b_{i}$. (It is even better to allow $d_{0}-1,\ldots,d$, but
the analysis gets a little messy.) The algorithm uses $T$ such buckets,
independently chosen. The probability of a point $x$ falling into a bucket is
\beq p_{A*}=\left(\barr{cc}d\\d_{0}-1\earr\right)2^{-d}+
\left(\barr{cc}d\\d_{0}\earr\right)2^{-d} \eeq
Let the number of points be
\beq n_{0}=n_{1}=n=\lfloor 1/p_{A*} \rfloor \eeq
This way the expected number of comparisons (point pairs in the same bucket) is
\beq T(np_{A*})^{2}\le T \eeq
The probability that both special pair points fall at least once into
the same bucket is
\beq S=\sum_{m=0}^{d}\left(\barr{cc}d\\m\earr\right)p^{d-m}
(1-p)^{m}\left[1-\left(1-S_{m}\right)^{T}\right] \eeq
\beq S_{m}=2^{-d}\left(\barr{cc}m\\ \lfloor m/2\rfloor\earr\right)
\left[\left(\barr{cc}d-m\\d_{0}-\lceil m/2\rceil\earr\right)+
\left(\barr{cc}d-m\\d_{0}-\lceil (m+1)/2\rceil\earr\right)\right] \eeq
The explanation follows. In these formulas $m$ is the number of coordinates
$i$ at which the special pair values disagree: $x_{0,i}\neq x_{1,i}$.
Consider the special pair fixed. There are $2^{d}$ possible baskets,
independently chosen. Consider one basket. For $j,k=0,1$ denote by $m_{jk}$ the
number of coordinates $i$ such that $x_{0,i}\oplus b_{i}=j$ and
$x_{0,i}\oplus x_{1,i}=k$ where $\oplus$ is the xor operation.
We know that $m_{01}+m_{11}=m$ and $m_{00}+m_{10}=d-m$.
Both $x_{0},x_{1}$ fall into the basket iff
$m_{00}+m_{01}=d_{0}-1,d_{0}$ and $m_{00}+m_{11}=d_{0}-1,d_{0}$.
There are two possibilities
\beq \left( \barr{cc}m_{00} & m_{01} \\m_{10} & m_{11}\earr \right)=
\left( \barr{cc}d_{0}-\lceil m/2\rceil & \lfloor m/2\rfloor \\
d-d_{0}-\lfloor m/2\rfloor & \lceil m/2\rceil \earr \right) \eeq
\beq \left( \barr{cc}m_{00} & m_{01} \\m_{10} & m_{11}\earr \right)=
\left( \barr{cc}d_{0}-\lceil (m+1)/2\rceil & \lceil m/2\rceil \\
d-d_{0}-\lfloor (m-1)/2\rfloor & \lfloor m/2\rfloor \earr \right) \eeq
each providing
\beq \left(\barr{cc}m_{00}+m_{10}\\m_{00}\earr\right)
\left(\barr{cc}m_{01}+m_{11}\\m_{01}\earr\right) \eeq
buckets.

Clearly $m$ obeys a Bernoulli$(1-p)$ distribution, so by
Chebyshev's inequality
\beq S\ge \min_{|m-(1-p)d|<\sqrt{p(1-p)d/\epsilon}}\left(1-e^{-TS_{m}}
-\epsilon \right) \eeq
for any $0<\epsilon<1$. Hence taking
\beq T=\lceil-\ln\epsilon/\min_{|m-(1-p)d|<\sqrt{p(1-p)d/\epsilon}}S_{m}
\rceil \eeq
guaranties a success probability $S\ge 1-2\epsilon$.
What is the relationship between $n$ and $T$? Let
\beq d_{0}\sim (1+\rho)d/2,\quad d\rightarrow\infty \eeq
By Stirling's approximation
\beq \lim\frac{\ln n}{d}=I\left(\frac{1+\rho}{2}\right) \eeq
\beq \lim\frac{\ln T}{d}=pI\left(\frac{1+\rho/p}{2}\right) \eeq
Letting $\rho\rightarrow 0$ results in exponent
\beq \lim\frac{\ln T}{\ln n}=\frac{1}{p} \eeq

We are not yet finished with this algorithm, because the number of comparisons
is not the only component of work. One also has to throw the points into
the baskets. The straightforward way of doing it is to check the point-basket
pairs. This involves $2nT$ checks, which is worse than the naive $n^2$
algorithm! In order to overcome this, we take the $k$'th tensor power
of the previous algorithm. That means throwing $n^{k}$ points in
$\{0,1\}^{kd}$ into $T^{k}$ buckets, by dividing the coordinates into
$k$ blocks of size $d$. The success probability is $S^{k}$, the expected
number of comparisons is at most $T^{k}$, but throwing the points into
the baskets takes only an expected number of $2n^{k}T$ vector operations
(of length $kd$). Hence the total expected number of vector operations
is at most
\beq T^{k}+2n^{k}T \eeq
At last taking
\beq k=\lceil 1/(1-p)\rceil \eeq
lets us approach the promised exponent $1/p$.

\section{The Probabilistic Model}

\begin{definition}
The pairwise independent identically distributed data model is the following.
Let the sets
\beq X_{0}\subset \{0,1,\ldots,b_{0}-1\}^{d},\quad
X_{1}\subset \{0,1,\ldots,b_{1}-1\}^{d} \eeq
of cardinalities
$\#X_{0}=n_{0},\ \ \#X_{1}=n_{1}$
be randomly constructed using the probability matrix
\beq P=\left(\barr{llll} p_{00}&p_{01}&\ldots&p_{0\ b_{1}-1}\\
p_{10}&p_{11}&\ldots&p_{0\ b_{1}-1}\\
\vdots&\vdots&\ddots&\vdots\\
p_{b_{0}-1\ 0}&p_{b_{0}-1\ 1}&\ldots&p_{b_{0}-1\ b_{1}-1}\\
\earr\right) \eeq
\beq p_{jk}\ge 0,\quad \sum_{j=0}^{b_{0}-1}\sum_{k=0}^{b_{1}-1}p_{jk}=1 \eeq
The $X_{0}$ points are identically distributed pairwise independent Bernoulli
random vectors, with
\beq p_{j*}=\sum_{k=0}^{b_{1}-1}p_{jk} \eeq
probability that coordinate $i$ has value $j$.
The probability of a single point $x\in X_{0}$ is
\beq p_{x*}=\prod_{i=1}^{d}p_{x_{i}*} \eeq
and the probability of a set $B_{0}\subset X_{0}$ is of course
\beq p_{B_{0}*}=\sum_{x\in B_{0}}p_{x*} \eeq
Similarly $X_{1}$ is governed by
$p_{*k}=\sum_{j=0}^{b_{0}-1}p_{jk}$
There is a special pair of $X_{0},X_{1}$ points, uniformly chosen out of the
$n_{0}n_{1}$ possibilities. For that pair the probability that their
$i$'th coordinates are $j,k$ is $p_{jk}$ and for
$x_{0}\in X_{0},\ x_{1}\in X_{1}$
\beq p_{x_{0}x_{1}}=\prod_{i=1}^{d}p_{x_{0,i}x_{1,i}} \eeq
\end{definition}

Coding and information theory were initially developed for a similar model
(with a probability vector instead of a probability matrix). Extension to
non-uniform matrices, a stationary model with coordinate dependency, or
continuous data is possible, as was done for coding and information theory.

\section{Comparison with the Indyk-Motwani Analysis}

The Indyk-Motwani paper \cite{IM98} introduces a metric based, worst case
analysis. In general no average work upper bound can replace a worst case
work upper bound, and the reverse holds for lower bounds. Still some
comparison is unavoidable. Let us consider the uniform marginally
Bernoulli(1/2) problem with $d\rightarrow\infty$. We saw that the classical
approach requires $W\approx n^{\log_{2}2/p}$, and have reduced it to
$W\approx n^{\epsilon+1/p}$. What is the Indyk-Motwani bound? The Hamming
distance between two random points is approximately $d/2$ (the ratio to $d$
tends to $1/2$ as $d$ grows, according to the law of large numbers). The
Hamming distance between two related points is approximately $(1-p)d$.
Hence the distance ratio is $c=1/(2-2p)$ and the Indyk-Motwani work is
\beq W\approx n^{1+1/c}=n^{3-2p} \eeq
It can be argued that the drop in performance is offset by the lack of pairwise
independence assumptions.
The $n^{\frac{2}{1+e^{-1/c}}}=n^{\frac{2}{1+e^{2p-2}}}$ lower bound
of Motwani, Naor and Panigrahy \cite{MNP06} is interesting, but increasing
it to $n^{1/p}$ seems a challenge.

Now let us consider a typical sparse bits matrix: for a small $\epsilon$ let
\beq P=\left( \barr{cc}
1-3\epsilon & \epsilon \\
\epsilon & \epsilon
\earr \right) \eeq
The standard bucketing approach is to arrange the coordinates randomly and
hash each point by its first $k$ $1$'ns, where
$ k\approx -\ln n/\ln 2\epsilon $.
The probability that two unrelated points fall into the same bucket is less
than $(2\epsilon)^{k}\approx 1/n$, so the expected work per try is
approximately $n$. The probability that the two related points
fall into the same basket is at least
\beq \left(\barr{cc}m\\k\earr\right)(1-3\epsilon)^{m-k}\epsilon^{k}=
\left(\barr{cc}m\\k\earr\right)(1-3\epsilon)^{m-k}(3\epsilon)^{k}\cdot3^{-k}
\eeq
for any $m\ge k$ (consider the first $m$ coordinates). Taking
$m\approx k/3\epsilon$ shows that the success probability per try is at least
approximately $3^{-k}\approx n^{\ln 3/\ln 2\epsilon}$. Hence in order to
succeed we will make $n^{-\ln 3/\ln 2\epsilon}$ tries, and the total expected
work is
\beq W\approx n^{1+\frac{\ln 3}{\ln 1/2\epsilon}} \eeq
In contrast the Hamming distance between random points is approximately
$2(1-2\epsilon)2\epsilon d$ and the Hamming distance between two related points
is approximately $2\epsilon d$, so the Indyk-Motwani distance ratio is
$c=2(1-2\epsilon)\approx 2$ and
\beq W\approx n^{1+1/c}\approx n^{3/2} \eeq
This worst case bound does not preclude the possibility that the random
projections approach recommended for sparse data by Datar Indyk Immorlica and
Mirrokni \cite{DIIM04} performs better. Their optimal choice
$r\rightarrow\infty$ results in a binary hash function
$ h(x)={\rm sign}\left(\sum_{i=1}^dx_{i}C_{i}\right) $
where $(x_{1},x_{2},\ldots,x_{d})\in X$ is a any point and
$C_{1},C_{2},\ldots,C_{d}$ are independent Cauchy random variables (density
$\frac{1}{\pi(1+z^{2})}$). Both $\pm 1$ values have probability
$1/2$, so one has to concatenate $k\approx\log_{2}n$ binary hash functions
in order to determine a bucket. Now consider two related points. They will have
approximately $\epsilon d$ $1$'ns in common, and each will have approximately
$\epsilon d$ $1$'ns where the other has zeroes. The sum of $\epsilon d$
independent Cauchy random variables has the same distribution as $\epsilon d$
times a single Cauchy random variable, so the probability that the two related
points get the same hash bit is approximately
\beq {\rm Prob}\left\{{\rm sign}\left(C_{1}+C_{2}\right)=
{\rm sign}\left(C_{1}+C_{3}\right)\right\}=2/3 \eeq
Hence amount of work is large:
\beq W\approx n(3/2)^{k}\approx n^{\log_{2}3} \eeq

We have demonstrated that the probabilistic model adds to the current
understanding of the approximate nearest neighbor problem. This is no surprise,
since it is the standard model of information theory.

\section{Bucketing Codes}

Assume that there is enough information to identify the special pair.
How much work is necessary? Comparing all $n_{0}n_{1}$ point pairs suffice.
All the effective known nearest neighbor algorithms are bucketing algorithms,
so will limit ourselves to these. But what are bucketing algorithms?
One could compute $m_{0},m_{1}$ in some complicated way from the data, and then
throw the $m_{0}$'th point of $X_{0}$ and the $m_{1}$'th point of $X_{1}$ into
a single bucket. It is unlikely to work, but can you prove it?
In order to disallow such knavery we will insist on data independent buckets.
Most practical bucketing algorithms are data dependent. That is necessary
because the data is used to construct (usually implicitly) a data model.
We suspect that when the data model is known, there is little to be gained
by making the buckets data dependent.

\begin{definition}
Assume the i.i.d. data model. A bucketing code is a set of $T$ subset pairs
\beqq (B_{0,0},B_{1,0}),\ldots,
(B_{0,T-1},B_{1,T-1})\subset X_{0}\times X_{1} \label{buckpair}\eeqq
Its success probability is
\beq S=p_{\cup_{t=0}^{T-1}B_{0,t}\times B_{1,t}} \eeq
and for any real numbers $n_{0},n_{1}>0$ its work is
\beqq W=\sum_{t=0}^{T-1}\max\left(n_{0}p_{B_{0,t}*},n_{1}p_{*B_{1,t}},
n_{0}p_{B_{0,t}*}n_{1}p_{*B_{0,t}}\right) \eeqq
\end{definition}

The meaning of success is obvious, but work has to be explained.
In the above definition we consider $n_{0},n_{1}$ to be the expected
number of $X_{0},X_{1}$ points, so they are not necessarily integers.
The simplest implementation of a bucketing code is to store it as two point
indexed arrays of lists. The first array of size $b_{0}^{d}$ keeps for each
point $x\in\{0,1,\ldots,b_{0}-1\}^{d}$ the list of buckets (from $0$ to $T-1$)
which contain it. The second array of size $b_{1}^{d}$ does the same for
the $B_{1,t}$'s.
When we are given $X_{0}$ and $X_{1}$ we look each element up, and accumulate
pointers to it in a buckets array of $k$ lists of pointers.
Then we compare the pairs in each of the $k$ buckets.
Let us count the expected number of operations. The expected number of buckets
containing any specific $X_{0}$ point is
$\sum_{t=0}^{T-1}p_{B_{0,t}*}$, so the $X_{0}$ lookup involves an
order of
$n_{0}+n_{0}\sum_{t=0}^{T-1}p_{B_{0,t}*}$
operations. Similarly the $X_{1}$ lookup takes
$n_{1}+n_{1}\sum_{t=0}^{T-1}p_{*B_{1,t}}$
The probability that a specific random pair falls into bucket $t$ is
$p_{B_{0,t}*}p_{*B_{1,t}}$, so the expected number of comparisons is
$n_{0}p_{B_{0,t}*}n_{1}p_{*B_{1,t}}$
It all adds up to
\beq n_{0}+n_{1}+\sum_{t=0}^{T-1}\big[n_{0}p_{B_{0,t}*}+
n_{1}p_{*B_{1,t}}+n_{0}p_{B_{0,t}*}n_{1}p_{*B_{1,t}}\big]
\le n_{0}+n_{1}+3W \eeq

The fly in the ointment is that for even moderate dimension $d$ the memory
requirements of the previous algorithm are out of the universe. Hence it can
be used only for small $d$. Higher dimensions can be handled by splitting
them up into short blocks, or by more sophisticated coding algorithms.

\section{Basic Results}

\begin{definition}
For any nonnegative matrix or vector $R$, and a probability matrix or vector
$P$ of the same dimensions $b_{0}\times b_{1}$, let the extended
Kullback-Leibler divergence be
\beq K(R\|P)=\sum_{j=0}^{b_{0}-1}\sum_{k=0}^{b_{1}-1}
r_{jk}\ln\frac{r_{jk}}{r_{**}p_{jk}}\ge 0 \eeq
where $r_{**}=\sum_{j=0}^{b_{0}-1}\sum_{k=0}^{b_{1}-1}r_{jk}$
\end{definition}

Non-negativity follows from the well known inequality:
\begin {lemma} \label{fun}
For any nonnegative $q_{0},q_{1},\ldots,q_{b-1}\ge 0$,
$p_{0},p_{1},\ldots,p_{b-1}\ge 0$
\beq \sum_{j=0}^{b-1}q_{j}\ln\frac{q_{j}}{p_{j}}\ge
q_{*}\ln\frac{q_{*}}{p_{*}} \eeq
where $q_{*}=\sum_{j=0}^{b-1}q_{j}\ ,\quad p_{*}=\sum_{j=0}^{b-1}p_{j}$
\end{lemma}

\begin{definition}
Suppose $P$ is a probability matrix. We write that
$\lambda_{0},\lambda_{1}\le 1\le\lambda_{0}+\lambda_{1}$ are $P$
$\mathbf{sub-conjugate}$ to each other, denoted by
$I(P,\lambda_{0},\lambda_{1},1)=0$, iff for any probability matrix $Q$
of the same dimensions as $P$
\beq K(Q_{\cdot\cdot}\|P_{\cdot\cdot})\ge
\lambda_{0}K(Q_{\cdot *}\|P_{\cdot *})+
\lambda_{1}K(Q_{* \cdot}\|P_{* \cdot}) \label{sub}\eeq
\end{definition}

Explicitly
\beq \sum_{j=0}^{b_{0}-1}\sum_{k=0}^{b_{1}-1}q_{jk}\ln\frac{q_{jk}}{p_{jk}}
\ge\lambda_{0}\sum_{j=0}^{b_{0}-1}q_{j*}\ln\frac{q_{j*}}{p_{j*}}+
\lambda_{1}\sum_{k=0}^{b_{1}-1}q_{*k}\ln\frac{q_{*k}}{p_{*k}} \eeq
where $q_{j*}=\sum_{k=0}^{b_{1}-1}q_{jk}$ etc.
The set of $P$ sub-conjugate pairs is convex by definition.

We will prove in the section \ref{book}
\begin{theorem} \label{direct}
For any bucketing code with probability matrix $P$, set sizes $n_{0},n_{1}$,
success probability $S$ and work $W$
\beq W\ge S\sup_{\lambda_{0},\lambda_{1}\le 1\le\lambda_{0}+\lambda_{1},
\ I(P,\lambda_{0},\lambda_{1},1)=0} n_{0}^{\lambda_{0}}n_{1}^{\lambda_{1}} \eeq
\end{theorem}

The following inverse result is a special case of theorem \ref{asy}
\begin{theorem} \label{inverse}
For any probability matrices $P,Q$, a scalar $\epsilon>0$ and large $N$ there
exists a bucketing code for matrix $P$, set sizes
$n_{0}=\lfloor N^{K(Q_{\cdot *}\|P_{\cdot *})}\rfloor$,
$n_{1}=\lfloor N^{K(Q_{* \cdot}\|P_{* \cdot})}\rfloor$,
with success probability $S\ge 1-\epsilon$ and work
$W\le N^{\epsilon + K(Q\|P)}$.
\end{theorem}

\section{An Example}

Consider the classical matrix
$P=\left( \barr{cc} p/2 & (1-p)/2 \\ (1-p)/2 & p/2 \earr \right)$.
Inserting $Q=\left( \barr{cc} 0 & 0 \\0 & 1 \earr \right)$
into theorem \ref{inverse} generates the well known
$n_{0}=n_{1}\approx N^{\ln 2}$,$S\ge 1-\epsilon$ and
$W\le N^{\epsilon + \ln 2/p}$.

The $Q\approx P$ neighborhood is important. Setting
$q_{jk}=p_{jk}+\delta_{jk}$, $\delta_{jk}\rightarrow 0$, $\delta_{**}=0$
results in $n_{0}\approx N^{\sum_{j}\frac{\delta_{j*}^{2}}{2p_{j*}}}$,
$n_{1}\approx N^{\sum_{k}\frac{\delta_{*k}^{2}}{2p_{*k}}}$, $S\ge 1-\epsilon$
and $W\le N^{\epsilon + \sum_{jk}\frac{\delta_{jk}^{2}}{2p_{jk}}}$.
Linear algebra shows that it is best to take
$\delta_{00}=-\delta_{11}=\delta$, $\delta_{10}=-\delta_{01}=\alpha\delta$.
Replacing $N$ with $N^{2/\delta^{2}}$ and $\epsilon$ with
$\epsilon\delta^{2}/2$ results in
$n_{0}\approx N^{(1-\alpha)^{2}}$, $n_{1}\approx N^{(1+\alpha)^{2}}$,
$S\ge 1-\epsilon$, $W\le N^{\epsilon+1/p+\alpha^{2}/(1-p)}$.
In particular for $\alpha=0$
$n_{0}=n_{1}=n$, $S\ge 1-\epsilon$, $W\le n^{\epsilon+1/p}$.

Is the exponent ${1/p}$ best possible?
Theorem \ref{direct} reduces the optimality of $1/p$ to a single inequality:
\begin{conjecture}
For any $1/2\le p\le 1$,
$q_{00},q_{01},q_{10},q_{11}\ge 0,q_{00}+q_{01}+q_{10}+q_{11}=1$
\berr 2p\Big[q_{00}\ln\frac{2q_{00}}{p}+q_{01}\ln\frac{2q_{01}}{1-p}+
q_{10}\ln\frac{2q_{10}}{1-p}+q_{11}\ln\frac{2q_{11}}{p}\Big]\ge\\ \ge
(q_{00}+q_{01})\ln 2(q_{00}+q_{01})+(q_{10}+q_{11})\ln 2(q_{10}+q_{11})+\\+
(q_{00}+q_{10})\ln 2(q_{00}+q_{10})+(q_{10}+q_{11})\ln 2(q_{10}+q_{11})
\eerr
\end{conjecture}

Computer experimentation and critical point analysis leave no doubt that
this inequality is valid. It is four dimensional, and keeping the marginal
probabilities fixed shows that we can further restrict
\beq (1-p)^{2}q_{00}q_{11}=p^{2}q_{01}q_{10} \eeq
A brute force proof is possible. Hopefully someone will find a clever proof.

Expressing $N,\alpha$ in terms of $n_{0},n_{1}$ shows that we can do with
$ e^{\frac{\ln n_{0}+\ln n_{1}-2(2p-1)\sqrt{\ln n_{0}\ln n_{1}}}
{4p(1-p)(1-\epsilon)}} $
comparisons. In particular when $n_{0}=n_{1}^{(2p-1)^{2}-\epsilon}$,
that asymmetric approximate nearest neighbor problem is solvable in
linear time!

\section{A Proof From The Book} \label{book}

In this section we will prove theorem \ref{direct}.

\begin{theorem}
For any probability matrices $P_{1},P_{2}$ and
$\lambda_{0},\lambda_{1}\le 1\le\lambda_{0}+\lambda_{1}$
\beq I(P_{1},\lambda_{0},\lambda_{1},1)=I(P_{2},\lambda_{0},\lambda_{1},1)=0
\iff I(P_{1}\times P_{2},\lambda_{0},\lambda_{1},1)=0 \eeq
where $\times$ is tensor product.
\end{theorem}
\begin{proof}
Direction $\Leftarrow$ is obvious, so assume the left hand side.
Denote $P=P_{1}\times P_{2}$:
\beq p_{j_{1}k_{1}j_{2}k_{2}}=p_{1,j_{1}k_{1}}p_{2,j_{2}k_{2}} \eeq
For any probability matrix
$\{q_{j_{1}k_{1}j_{2}k_{2}}\}_{j_{1}k_{1}j_{2}k_{2}}$
\beq \sum_{j_{1}k_{1}j_{2}k_{2}}q_{j_{1}k_{1}j_{2}k_{2}}\ln\frac
{q_{j_{1}k_{1}j_{2}k_{2}}}{p_{1,j_{1}k_{1}}p_{2,j_{2}k_{2}}}=
\sum_{j_{1}k_{1}}q_{j_{1}k_{1}**}\ln\frac{q_{j_{1}k_{1}**}}{p_{1,j_{1}k_{1}}}+
\sum_{j_{1}k_{1}j_{2}k_{2}}q_{j_{1}k_{1}j_{2}k_{2}}
\ln\frac{q_{j_{1}k_{1}j_{2}k_{2}}}{q_{j_{1}k_{1}**}p_{2,j_{2}k_{2}}} \eeq
Because  $I(P_{1},\lambda_{0},\lambda_{1},1)=0$
\beq \sum_{j_{1}k_{1}}q_{j_{1}k_{1}**}\ln\frac{q_{j_{1}k_{1}**}}
{p_{1,j_{1}k_{1}}} \ge
\lambda_{0}\sum_{j_{1}}q_{j_{1}***}\ln\frac{q_{j_{1}***}}{p_{1,j_{1}*}}+
\lambda_{1}\sum_{k_{1}}q_{*k_{1}**}\ln\frac{q_{*k_{1}**}}{p_{1,*k_{1}}} \eeq
Because  $I(P_{2},\lambda_{0},\lambda_{1},1)=0$
\berr \sum_{j_{2}k_{2}}q_{j_{1}k_{1}j_{2}k_{2}}/q_{j_{1}k_{1}**}\ln\frac
{q_{j_{1}k_{1}j_{2}k_{2}}/q_{j_{1}k_{1}**}}{p_{2,j_{2}k_{2}}} \ge \\ \ge
\lambda_{0}\sum_{j_{2}}q_{j_{1}k_{1}j_{2}*}/q_{j_{1}k_{1}**}\ln\frac
{q_{j_{1}k_{1}j_{2}*}/q_{j_{1}k_{1}**}}{p_{2,j_{2}*}}+
\lambda_{1}\sum_{k_{2}}q_{j_{1}k_{1}*k_{2}}/q_{j_{1}k_{1}**}\ln\frac
{q_{j_{1}k_{1}*k_{2}}/q_{j_{1}k_{1}**}}{p_{2,*k_{2}}} \eerr
\beq \sum_{j_{1}k_{1}j_{2}k_{2}}q_{j_{1}k_{1}j_{2}k_{2}}
\ln\frac{q_{j_{1}k_{1}j_{2}k_{2}}}{q_{j_{1}k_{1}**}p_{2,j_{2}k_{2}}} \ge
\lambda_{0}\sum_{j_{1}k_{1}j_{2}}q_{j_{1}k_{1}j_{2}*}\ln\frac
{q_{j_{1}k_{1}j_{2}*}}{q_{j_{1}k_{1}**}p_{2,j_{2}*}}+
\lambda_{1}\sum_{j_{1}k_{1}k_{2}}q_{j_{1}k_{1}*k_{2}}\ln\frac
{q_{j_{1}k_{1}*k_{2}}}{q_{j_{1}k_{1}**}p_{2,*k_{2}}} \eeq
so with help from lemma \ref{fun}
\beq \sum_{j_{1}k_{1}j_{2}k_{2}}q_{j_{1}k_{1}j_{2}k_{2}}
\ln\frac{q_{j_{1}k_{1}j_{2}k_{2}}}{q_{j_{1}k_{1}**}p_{2,j_{2}k_{2}}} \ge
\lambda_{0}\sum_{j_{1}j_{2}}q_{j_{1}*j_{2}*}\ln\frac
{q_{j_{1}*j_{2}*}}{q_{j_{1}***}p_{2,j_{2}*}}+
\lambda_{1}\sum_{k_{1}k_{2}}q_{*k_{1}*k_{2}}\ln\frac
{q_{*k_{1}*k_{2}}}{q_{*k_{1}**}p_{2,*k_{2}}} \eeq
Together
\beq \sum_{j_{1}k_{1}j_{2}k_{2}}q_{j_{1}k_{1}j_{2}k_{2}}\ln\frac
{q_{j_{1}k_{1}j_{2}k_{2}}}{p_{1,j_{1}k_{1}}p_{2,j_{2}k_{2}}}\ge
\lambda_{0}\sum_{j_{1}j_{2}}q_{j_{1}*j_{2}*}\ln\frac{q_{j_{1}*j_{2}*}}
{p_{1,j_{1}*}p_{2,j_{2}*}}+
\lambda_{1}\sum_{k_{1}k_{2}}q_{*k_{1}*k_{2}}\ln\frac{q_{*k_{1}*k_{2}}}
{p_{1,*k_{1}}p_{2,*k_{2}}} \eeq
hence $I(P_{1}\times P_{2},\lambda_{0},\lambda_{1},1)=0$.
\end{proof}

\begin{theorem}
For any $ B_{0}\subset \{0,1,\ldots,b_{0}-1\}^{d},\
B_{1}\subset \{0,1,\ldots,b_{1}-1\}^{d} $
\beq p_{B_{0}B_{1}}\le \min_{\lambda_{0},\lambda_{1}\le 1\le\lambda_{0}+
\lambda_{1},\ I(P,\lambda_{0},\lambda_{1},0)=0}
p_{B_{0}*}^{\lambda_{0}}p_{*B_{1}}^{\lambda_{1}} \eeq
\end{theorem}
\begin{proof}
Without restricting generality let $d=1$. Inserting
\beq q_{jk}=\left\{ \barr{ll}\frac{p_{jk}}{p_{B_{0}B_{1}}} &
j\in B_{0},k\in B_{1}\\ 0 & {\rm otherwise} \earr \right.  \eeq
into (\ref{sub}) proves the assertion.
\end{proof}

Proof of {\bf theorem \ref{direct}}.
\begin{proof}
Recall that the work is $W=\sum_{i}W_{i}$ where
\beq W_{i}=\max\left(n_{0}p_{B_{0,i}*},\ n_{1}p_{*B_{1,i}},
\ n_{0}p_{B_{0,i}*}n_{1}p_{*B_{0,i}}\right) \eeq
Our parameters satisfy
\beq (\lambda_{0},\lambda_{1})\in{\rm Conv}(\{(1,0),(0,1),(1,1)\}) \eeq
hence
\beq \ln W_{i}\ge \lambda_{0}\ln(n_{0}p_{B_{0,i}*})+
\lambda_{1}\ln(n_{1}p_{*B_{1,i}}) \eeq
\beq W_{i}\ge n_{0}^{\lambda_{0}}n_{1}^{\lambda_{1}}p_{B_{0,i}B_{1,i}} \eeq
Now sum up.
\end{proof}

\section{Bucketing Information}

All the results of this section will be proven in appendix \ref{infoproofs}.
\begin{definition} \label{buck}
Suppose $P$ is a probability matrix. The $\mathbf{bucketing}$
$\mathbf{information}$ function is for $\mu\ge 0$
\berrr I(P,\lambda_{0},\lambda_{1},\mu)=\max_{\footnotesize\barr{ll}
\{r_{i,jk}\ge 0\}_
{\barr{lll}0\le i<b_{0}b_{1}\\0\le j<b_{0} \\0\le k<b_{1}\earr}
\\ r_{*,**}=1\earr} \Bigg[
\lambda_{0}\sum_{i=0}^{b_{0}b_{1}-1}K(R_{i,\cdot *}\|P_{\cdot *})+ 
\lambda_{1}\sum_{i=0}^{b_{0}b_{1}-1}K(R_{i,*\cdot}\|P_{*\cdot})+\\+
(1-\mu)K(R_{*,\cdot\cdot}\|P_{\cdot\cdot})-
\sum_{i=0}^{b_{0}b_{1}-1}K(R_{i,\cdot\cdot}\|P_{\cdot\cdot}) \Bigg]\eerrr
Explicitly$ r_{i,j*}=\sum_{k=0}^{b_{1}-1}r_{i,jk} $,
$ K(R_{i,\cdot *}\|P_{\cdot *})=\sum_{j=0}^{b_{0}-1}r_{i,j*}
\ln\frac{r_{i,j*}}{r_{i,**}p_{j*}} $ etc.
\end{definition}

\begin{lemma} \label{equivalent}
For any probability matrix $P$ and $0\le\mu\le 1$ the sums in definition
\ref{buck} can be restricted to a single term, i.e.
\beq I(P,\lambda_{0},\lambda_{1},\mu)=\max_{Q}\Bigg[
\lambda_{0}K(Q_{\cdot *}\|P_{\cdot *})+ 
\lambda_{1}K(Q_{*\cdot}\|P_{*\cdot})-
\mu K(Q_{\cdot\cdot}\|P_{\cdot\cdot}) \Bigg] \label{single}\eeq
where $Q$ ranges over all probability matrices.
For any $\mu\ge 0$, not restricting the number of terms $i$ in definition
\ref{buck} does not change $I$. It can be rewritten as
\beq I(P,\lambda_{0},\lambda_{1},\mu)=\max_{Q}\bigg[(1-\mu)
K(Q_{\cdot\cdot}\|P_{\cdot\cdot})+\max_{(Q,y)\in {\rm Conv}\left(
G(P,\lambda_{0},\lambda_{1})\right) }y \bigg]  \label{def2}\eeq
where ${\rm Conv}$ is the convex hull and
\beq G(P,\lambda_{0},\lambda_{1})=\{(Q,\ 
\lambda_{0}K(Q_{\cdot *}\|P_{\cdot *})+\lambda_{1}K(Q_{*\cdot}\|P_{*\cdot})-
K(Q_{\cdot\cdot}\|P_{\cdot\cdot}))\}_{Q} \eeq
\end{lemma}

From now on when dealing with the bucketing information function, we will
denote $\sum_{i}$ without worrying about the number of indices.

\begin{lemma} \label{special}
For any probability matrix $P$ and $\mu\ge 0$ the bucketing information
function $I(P,\lambda_{0},\lambda_{1},\mu)$
is nonnegative, convex, monotonically nondecreasing in
$\lambda_{0},\lambda_{1}$ and monotonically non-increasing in $\mu$.
Special values are
\beq I(P,\lambda_{0},\lambda_{1},\mu)=\mu I(P,\lambda_{0}/\mu,\lambda_{1}/\mu,
1)\quad 0<\mu \le 1 \label{obv}\eeq
\beq I(P,\lambda_{0},\lambda_{1},\mu)=0 \iff
\forall Q ,\ \min(\mu,1)K(Q_{\cdot\cdot}\|P_{\cdot\cdot}))\ge
\lambda_{0}K(Q_{\cdot *}\|P_{\cdot *})+
\lambda_{1}K(Q_{*\cdot}\|P_{*\cdot})  \label{izero}\eeq
\beq I(P,\lambda_{0},\lambda_{1},\mu)=0\qquad 0\le\lambda_{0},\lambda_{1}
\ \ \lambda_{0}+\lambda_{1}\le\min(\mu,1) \label{izero2}\eeq
\beq I(P,1,1,\mu)= \max_{0\le j<b_{0},0\le k<b_{1}}\ln\frac{(p_{jk})^{\mu}}
{p_{j*}p_{*k}} \qquad 0\le\mu\le 1 \eeq
\beq I(P,1,1,\mu)=(\mu-1)\ln\sum_{j=0}^{b_{0}-1}\sum_{k=0}^{b_{1}-1}
p_{jk}\left(\frac{p_{jk}}{p_{j*}p_{*k}}\right)^{\frac{1}{\mu-1}}
\quad \mu\ge 1 \eeq
\beq I(P,1,1,\infty)=I(P)=\sum_{j=0}^{b_{0}-1}\sum_{k=0}^{b_{1}-1}
p_{jk}\ln\frac{p_{jk}}{p_{j*}p_{*k}} \eeq
\end{lemma} 

\begin{theorem} \label{indy}
For any probability matrices $P_{1},P_{2}$ and $\mu\ge 0$
\beq I(P_{1}\times P_{2},\lambda_{0},\lambda_{1},\mu)=
I(P_{1},\lambda_{0},\lambda_{1},\mu)+I(P_{2},\lambda_{0},\lambda_{1},\mu) \eeq
\end{theorem}

\section{Bucketing Codes and Information}

All the results of this section will be proven in appendix \ref{codesproofs}.
\begin{theorem} \label{big}
For any bucketing code with probability matrix
$P_{1}\times P_{2}\times\cdots\times P_{\tilde{d}}$, dimension $d=1$,
set sizes $n_{0},n_{1}$, success probability $S$ and work $W$
\beq \ln W\ge \sup_{\lambda_{0},\lambda_{1}\le 1\le
\lambda_{0}+\lambda_{1},\ \mu\ge 0}
\left[\lambda_{0}\ln n_{0}+\lambda_{1}\ln n_{1}+\mu\ln S-\sum_{i=1}^{\tilde{d}}
I(P_{i},\lambda_{0},\lambda_{1},\mu) \right] \eeq
\end{theorem}

\begin{definition}
Assume the i.i.d. data model with probability matrix $P$.
Suppose there exists a $d$ dimensional bucketing code
such that for the expected numbers $n_{0},n_{1}$
of $X_{0},X_{1}$ points it has success probability $S$ and work $W$.
Then for any real numbers $0\le\tilde{S}\le S,\ \tilde{W}\ge W$ we say that
$(P,d,n_{0},n_{1},\tilde{S},\tilde{W})$
is $\mathbf{attainable}$.
Define the set of $\mathbf{log-attainable}$ parameters to be
\beq D(P)=\bigg\{\left.\frac{1}{d}(\ln n_{0},\ln n_{1},-\ln S,\ln W)\ 
\right| (P,d,n_{0},n_{1},S,W) {\rm\ is\ attainable} \bigg\} \eeq
Normalizing by $d$ is awkward in the infinite data case $d=\infty$.
There it makes sense to consider the $\mathbf{log-attainable}$ $\mathbf{cone}$
\beq D_{0}(P)={\rm Cone}(D(P))=\cup_{\alpha\ge 0}\alpha D(P) \eeq
\end{definition}

Theorem \ref{big} is asymptotically tight in the following sense:
\begin{theorem} \label{asy}
For any probability matrix $P$ the closure of its log-attainable set is
\berr D^{c}(P)\ =\ \{(m_{0},m_{1},s,w)\ |\ s\ge 0, \label{clos}\\ \forall\ 
\lambda_{0},\lambda_{1}\le 1\le\lambda_{0}+\lambda_{1},\ \mu\ge 0\quad
w\ge\lambda_{0}m_{0}+\lambda_{1}m_{1}-
\mu s-I(P,\lambda_{0},\lambda_{1},\mu) \} \eerr
Equivalently
\berr D^{c}(P)\ =\ D(0)+{\rm Conv}\bigg(\Big\{\Big(
\sum_{i}K(R_{i,\cdot *}\|P_{\cdot *}),\  
\sum_{i}K(R_{i,*\cdot}\|P_{*\cdot}),\ K(R_{*,\cdot\cdot}\|P_{\cdot\cdot}),
\label{conv} \\ -K(R_{*,\cdot\cdot}\|P_{\cdot\cdot})+
\sum_{i}K(R_{i,\cdot\cdot}\|P_{\cdot\cdot})
\Big)\Big\}_{r_{i,jk}\ge 0,\ r_{*,**}=1} \bigg) \eerr
where $D(0)$ is the common core
\beq D(0)={\rm ConvCone}(\{(1,0,0,1),(0,1,0,1),
(0,0,1,0),(0,0,0,1),(-1,-1,0,-1)\}) \eeq

For the unlimited data case $d\rightarrow\infty$ 
\berr D_{0}^{c}(P)=\{(m_{0},m_{1},s,w)\ |\ s\ge 0\}\cap \label{unli} \\ \cap
\big[D_{0}(0)+{\rm ConvCone}(\{(
K(Q_{\cdot *}\|P_{\cdot *}),\ K(Q_{*\cdot}\|P_{*\cdot}),\
\ K(Q_{\cdot\cdot}\|P_{\cdot\cdot}),\ 0)\}_{Q})\big] \eerr
where $D_{0}(0)$ is the extended common core
\beq D_{0}(0)=D(0)+{\rm Cone}(\{(0,0,-1,1)\}) \eeq
and $Q$ runs over all $b_{0}\times b_{1}$ probability matrices.
\end{theorem}

In light of theorem \ref{asy}, theorem \ref{indy} can be recast as
\begin{theorem}
For any probability matrices $P_{1},P_{2}\ $
$ D^{c}(P_{1}\times P_{2})=D^{c}(P_{1})+D^{c}(P_{2}) $
\end{theorem}

\section{Conclusion}

We consider the approximate nearest neighbor problem in a probabilistic
setting. Using several coordinates at once enables asymptotically better
approximate nearest neighbor algorithms than using them one at a time.
The performance is bounded by, and tends to, a newly defined bucketing
information function. Thus bucketing coding and
information theory play the same role for the approximate nearest
neighbor problem that Shannon's coding and information theory play
for communication.

%\ifCLASSOPTIONcaptionsoff
%  \newpage
%\fi



%\bibliographystyle{IEEEtran}

\begin{thebibliography}{1}

\bibitem{NA}
N.~Alon
Private Communication.

\bibitem{AI06}
A.~Andoni, P.~Indyk
{\em Near-Optimal Hashing Algorithms for Approximate Nearest Neighbor in
High Dimensions}
FOCS 2006.

\bibitem{Bro98}
A.~Broder.
{\em Identifying and Filtering Near-Duplicate Documents}
Proc. FUN, 1998.

\bibitem{DIIM04}
M.~Datar, P.~Indyk, N.~Immorlica and V.~Mirrokni
{\em Locality-Sensitive Hashing Scheme Based on p-Stable Distributions}
Proc. Sympos. on Computational Geometry, 2004.

\bibitem{GSZ01}
C.Gennaro, P.Savino and P.Zezula
{\em Similarity Search in Metric Databases through Hashing}
Proc. ACM workshop on multimedia, 2001.

\bibitem{IM98}
P.~Indyk and R.~Motwani.
{\em Approximate Nearest Neighbor: Towards Removing the Curse
of Dimensionality}
Proc. 30th Annu. ACM Sympos. Theory Comput., 1998.

\bibitem{KWZ95}
R.M. Karp, O. Waarts, and G. Zweig.
{\em The Bit Vector Intersection Problem}
Proc. 36th Annu. IEEE Sympos. Foundations of Computer Science,
pp. 621-630, 1995.

\bibitem{MNP06}
R. Motwani, A. Naor and R. Panigrahy
{\em Lower Bounds on Locality Sensitive Hashing}
SCG'06

\bibitem{MD1}
{\em A Heterogeneous High Dimensional Approximate Nearest
Neighbor Algorithm}
To be Published.

\end{thebibliography}

\appendices

\section{Bucketing Information Proofs} \label{infoproofs}

Proof of {\bf Lemma \ref{equivalent}}.
\begin{proof}
Lemma (\ref{fun}) implies that
\beq K(R_{*,\cdot\cdot}\| P_{\cdot\cdot})\le
\sum_{i=0}^{b_{0}b_{1}-1}K(R_{i,\cdot\cdot}\| P_{\cdot\cdot}) \label{again}\eeq
so for $0\le\mu\le 1$
\beq (1-\mu)K(R_{*,\cdot\cdot}\|P_{\cdot\cdot})-
\sum_{i=0}^{b_{0}b_{1}-1}K(R_{i,\cdot\cdot}\|P_{\cdot\cdot}) \le
-\mu\sum_{i=0}^{b_{0}b_{1}-1}K(R_{i,\cdot\cdot}\|P_{\cdot\cdot}) \eeq
and only one $i$ is necessary.
The connection between definition \ref{buck} and (\ref{def2}) is through
$ r_{i}=r_{i,**} $ ,
$ q_{i,jk}=\frac{r_{i,jk}}{r_{i,**}} $
\berr I(P,\lambda_{0},\lambda_{1},\mu)=\max_{\footnotesize
\barr{cc}\{r_{i},Q_{i}\}_{i} \\
r_{*}=1\earr} \Bigg[\sum_{i=0}^{b_{0}b_{1}-1} r_{i}\Big[
\lambda_{0}K(Q_{i,\cdot *}\|P_{\cdot *})+
\lambda_{1}K(Q_{i,*\cdot}\|P_{*\cdot})+ \\ +
(1-\mu)K\Big(\sum_{i}r_{i}Q_{i,\cdot\cdot}\Big\| P_{\cdot\cdot}\Big)
-K(Q_{i,\cdot\cdot}\|P_{\cdot\cdot}) \Big] \Bigg] \eerr
The set $G$ is $b_{0}b_{1}$ dimensional, so by
Caratheodory's theorem any point on the boundary of its convex hull is a 
convex combination of $b_{0}b_{1}$ $G$ points.
\end{proof}

Proof of {\bf lemma \ref{special}}.
\begin{proof}
Non-negativity follows by taking $Q=P$. Monotonicity ,convexity
and (\ref{obv}) are by definition.

When $0\le\mu\le 1$  (\ref{single}) is valid and (\ref{izero}) is clear.
When $\mu\ge 1$
\beq I(P,\lambda_{0},\lambda_{1},\mu)\le\max_{R}\sum_{i=0}^{b_{0}b_{1}-1}
\Bigg[\lambda_{0}K(R_{i,\cdot *}\|P_{\cdot *})+ 
\lambda_{1}K(R_{i,*\cdot}\|P_{*\cdot})-
K(R_{i,\cdot\cdot}\|P_{\cdot\cdot}) \Bigg]\eeq
so direction $\Leftarrow$ of (\ref{izero}) is true.
On the other hand assume that for some $Q$
\beq K(Q_{\cdot\cdot}\|P_{\cdot\cdot})<
\lambda_{0}K(Q_{\cdot *}\|P_{\cdot *}) +
\lambda_{1}K(Q_{*\cdot}\|P_{*\cdot}) \eeq
Inserting $r_{0,jk}=\epsilon q_{jk},\ r_{1,jk}=p_{jk}-\epsilon q_{jk}$
into definition \ref{buck} gives
\berr I(P,\lambda_{0},\lambda_{1},\mu)\ge
\epsilon\left[\lambda_{0}K(Q_{\cdot *}\|P_{\cdot *}) +
\lambda_{1}K(Q_{*\cdot}\|P_{*\cdot}) - K(Q_{\cdot\cdot}\|P_{\cdot\cdot})
\right] + \\ +
(1-\epsilon)\left[\lambda_{0}K(\tilde{P}_{\cdot *}\|P_{\cdot *}) +
\lambda_{1}K(\tilde{P}_{*\cdot}\|P_{*\cdot}) - 
K(\tilde{P}_{\cdot\cdot}\|P_{\cdot\cdot}) \right] \eerr
where $\tilde{P}=(P-\epsilon Q)/(1-\epsilon)=P+\epsilon(P-Q)/(1-\epsilon)$.
The Kullback-Leibler divergence between $\tilde{P}$ and $P$ is second order
in $\epsilon$, and the same holds for their marginal vectors. Hence for a small
$\epsilon>0$ $I(P,\lambda_{0},\lambda_{1},\mu)>0$, and the proof of
(\ref{izero}) is done.

Lemma \ref{fun} implies
\beq K(R_{i,\cdot *}\|P_{\cdot*}),K(R_{i,*\cdot}\|P_{*\cdot})\le
K(R_{i,\cdot\cdot}\|P_{\cdot\cdot}) \eeq
so (\ref{izero2}) follows from (\ref{izero}).

Now to $\lambda_{0}=\lambda_{1}=1$. We want to maximize
\beqq \sum_{i}\big[K(R_{i,\cdot *}\|P_{\cdot *})+K(R_{i,*\cdot}\|P_{*\cdot})
-K(R_{i,\cdot\cdot}\|P_{\cdot\cdot})\big]=
\sum_{jk}r_{*,jk}\ln\frac{p_{jk}}{p_{j*}p_{*k}}-
\sum_{ijk}r_{i,jk}\ln\frac{r_{i,**}r_{i,jk}}{r_{i,j*}r_{i,*k}} \eeqq
The rightmost sum is nonnegative, and for any $\{r_{*,jk}\}_{jk}$ it can be
made $0$ by choosing
\beq r_{i,jk}=\left\{\barr{ll}r_{*,jk}&i=j+b_{0}k\\0&{\rm otherwise}\earr
\right. \eeq
Hence we want to maximize
\beq \sum_{jk}r_{*,jk}\ln\frac{(p_{jk})^{\mu}}{p_{j*}p_{*k}}+
(1-\mu)\sum_{jk}r_{*,jk}\ln r_{*,jk} \eeq
When $0\le\mu\le 1$ both sums can be simultaneously maximized by concentrating
$r$ in one place.
When $\mu\ge 1$ the maximized function is concave in $\{r_{*,jk}\}_{jk}$,
and Lagrange multipliers reveal the optimal choice
\beq r_{*,jk}=\frac{\left(\frac{(p_{jk})^{\mu}}{p_{j*}p_{*k}}\right)^
{\frac{1}{\mu-1}}}{\sum_{\tilde{j}\tilde{k}}
\left(\frac{(p_{\tilde{j}\tilde{k}})^{\mu}}{p_{\tilde{j}*}p_{*\tilde{k}}}
\right)^{\frac{1}{\mu-1}}}\eeq
\end{proof}

Proof of {\bf theorem \ref{indy}}.
\begin{proof}
Obviously $I(P_{1}\times P_{2},\lambda_{0},\lambda_{1},\mu)\ge
I(P_{1},\lambda_{0},\lambda_{1},\mu)+I(P_{2},\lambda_{0},\lambda_{1},\mu)$.
The other direction is the challenge. Denote $P=P_{1}\times P_{2}$:
\beq p_{j_{1}k_{1}j_{2}k_{2}}=p_{1,j_{1}k_{1}}p_{2,j_{2}k_{2}} \eeq
For any $\{r_{i,j_{1}j_{2}k_{1}k_{2}}\}_{i,j_{1}j_{2}k_{1}k_{2}}$
\berrr (\mu-1)K(R_{*,\cdot\cdot\cdot\cdot}\|P_{\cdot\cdot\cdot\cdot})+
\sum_{i}K(R_{i,\cdot\cdot\cdot\cdot}\|P_{\cdot\cdot\cdot\cdot})=\\=
(\mu-1)\sum_{j_{1}k_{1}}r_{*,j_{1}k_{1}**}\ln\frac{r_{*,j_{1}k_{1}**}}
{p_{1,j_{1}k_{1}}}+
\sum_{i,j_{1}k_{1}}r_{i,j_{1}k_{1}**}\ln\frac{r_{i,j_{1}k_{1}**}}
{r_{i,****}p_{1,j_{1}k_{1}}}+\\+
(\mu-1)\sum_{j_{1}k_{1}j_{2}k_{2}}r_{*,j_{1}k_{1}j_{2}k_{2}}
\ln\frac{r_{*,j_{1}k_{1}j_{2}k_{2}}}{r_{*,j_{1}k_{1}**}p_{2,j_{2}k_{2}}}+
\sum_{i,j_{1}k_{1}j_{2}k_{2}}r_{i,j_{1}k_{1}j_{2}k_{2}}
\ln\frac{r_{i,j_{1}k_{1}j_{2}k_{2}}}{r_{i,j_{1}k_{1}**}p_{2,j_{2}k_{2}}}
\eerrr
By definition
\berrr (\mu-1)\sum_{j_{1}k_{1}}r_{*,j_{1}k_{1}**}\ln\frac{r_{*,j_{1}k_{1}**}}
{p_{1,j_{1}k_{1}}}+
\sum_{i,j_{1}k_{1}}r_{i,j_{1}k_{1}**}\ln\frac{r_{i,j_{1}k_{1}**}}
{r_{i,****}p_{1,j_{1}k_{1}}} \ge \\ \ge
\lambda_{0}\sum_{i,j_{1}}r_{i,j_{1}***}\ln\frac{r_{i,j_{1}***}}
{r_{i,****}p_{1,j_{1}*}}+
\lambda_{1}\sum_{i,k_{1}}r_{i,*k_{1}**}\ln\frac{r_{i,*k_{1}**}}
{r_{i,****}p_{1,*k_{1}}} - I(P_{1},\lambda_{0},\lambda_{1},\mu) \eerrr
\berrr (\mu-1)\sum_{j_{2}k_{2}}r_{*,j_{1}k_{1}j_{2}k_{2}}\ln\frac
{r_{*,j_{1}k_{1}j_{2}k_{2}}}{r_{*,j_{1}k_{1}**}p_{2,j_{2}k_{2}}}+
\sum_{i,j_{2}k_{2}}r_{i,j_{1}k_{1}j_{2}k_{2}}\ln\frac
{r_{i,j_{1}k_{1}j_{2}k_{2}}}{r_{i,j_{1},k_{1}**}p_{2,j_{2}k_{2}}} \ge \\ \ge
\lambda_{0}\sum_{i,j_{2}}r_{i,j_{1}k_{1}j_{2}*}\ln\frac
{r_{i,j_{1}k_{1}j_{2}*}}{r_{i,j_{1}k_{1}**}p_{2,j_{2}*}}+
\lambda_{1}\sum_{i,k_{2}}r_{i,j_{1}k_{1}*k_{2}}\ln\frac
{r_{i,j_{1}k_{1}*k_{2}}}{r_{i,j_{1},k_{1}**}p_{2,*k_{2}}} -\\-
r_{*,j_{1}k_{1}**}I(P_{2},\lambda_{0},\lambda_{1},\mu) \eerrr
so with help from lemma \ref{fun}
\berrr (\mu-1)\sum_{j_{1}k_{1}j_{2}k_{2}}r_{*,j_{1}k_{1}j_{2}k_{2}}\ln\frac
{r_{*,j_{1}k_{1}j_{2}k_{2}}}{r_{*,j_{1}k_{1}**}p_{2,j_{2}k_{2}}}+
\sum_{i,j_{1}k_{1}j_{2}k_{2}}r_{i,j_{1}k_{1}j_{2}k_{2}}\ln\frac
{r_{i,j_{1}k_{1}j_{2}k_{2}}}{r_{i,j_{1},k_{1}**}p_{2,j_{2}k_{2}}} \ge \\ \ge
\lambda_{0}\sum_{i,j_{1}j_{2}}r_{i,j_{1}*j_{2}*}\ln\frac
{r_{i,j_{1}*j_{2}*}}{r_{i,j_{1}***}p_{2,j_{2}*}}+
\lambda_{1}\sum_{i,k_{1}k_{2}}r_{i,*k_{1}*k_{2}}\ln\frac
{r_{i,*k_{1}*k_{2}}}{r_{i,*,k_{1}**}p_{2,*k_{2}}} -
I(P_{2},\lambda_{0},\lambda_{1},\mu) \eerrr
Together
\berrr (\mu-1)K(R_{*,\cdot\cdot\cdot\cdot}\|P_{\cdot\cdot\cdot\cdot})+
\sum_{i}K(R_{i,\cdot\cdot\cdot\cdot}\|P_{\cdot\cdot\cdot\cdot}) \ge \\ \ge
\lambda_{0}\sum_{i}K(R_{i,\cdot *\cdot *}\|P_{\cdot *\cdot *})+
\lambda_{1}\sum_{i}K(R_{i,*\cdot *\cdot}\|P_{*\cdot *\cdot})-
I(P_{1},\lambda_{0},\lambda_{1},\mu)-I(P_{2},\lambda_{0},\lambda_{1},\mu)
\eerrr
Notice that we have used the fact that for $0\le\mu\le 1$ there is only
one $i$.
\end{proof}

\section{Bucketing Codes and Information Proofs} \label{codesproofs}

Proof of {\bf theorem \ref{big}}.
\begin{proof}
Without restricting generality let $\nu=1$. Let
$(B_{0,0},B_{1,0}),\cdots$,$(B_{0,T-1},B_{1,T-1})$ be subset pairs. Denote
\beq B_{i}=B_{0,i}\times B_{1,i}\ \backslash\ \bigcup_{t=0}^{i-1}
B_{0,t}\times B_{1,t} \eeq
so the success probability is
$ S=\sum_{i}p_{B_{i}} $
Insert
\beq r_{i,jk}=\left\{ \barr{ll}\frac{p_{jk}}{S} &
(j,k)\in B_{i}\\ 0 & {\rm otherwise} \earr \right.  \eeq
into definition \ref{buck}. Lemma \ref{fun} implies
\beq K(R_{i,\cdot *}\|P_{\cdot *})=\sum_{j\in B_{0,i}}r_{i,j*}\ln\frac
{r_{i,j*}}{r_{i,**}p_{j*}} \ge -r_{i,**}\ln p_{B_{0,i}*} \eeq
Similarly
\beq K(R_{i,*\cdot}\|P_{\cdot *}) \ge -r_{i,**}\ln p_{*B_{1,i}} \eeq
\beq \sum_{i}\left[\lambda_{0}K(R_{i,\cdot *}\|P_{\cdot *})+
\lambda_{1}K(R_{i,*\cdot}\|P_{*\cdot})
\right] \ge -\sum_{i}r_{i,**}\left(\lambda_{0}\ln p_{B_{0,i}*}+
\lambda_{1}\ln p_{*B_{1,i}} \right) \eeq
Recall that the work is $W=\sum_{i}W_{i}$ where
\beq W_{i}=\max\left(n_{0}p_{B_{0,i}*},\ n_{1}p_{*B_{1,i}},
\ n_{0}p_{B_{0,i}*}n_{1}p_{*B_{0,i}}\right) \eeq
Our parameters satisfy
\beq (\lambda_{0},\lambda_{1})\in{\rm Conv}(\{(1,0),(0,1),(1,1)\}) \eeq
hence
\beq \ln W_{i}\ge \lambda_{0}\ln(n_{0}p_{B_{0,i}*})+
\lambda_{1}\ln(n_{1}p_{*B_{1,i}}) \eeq
\beq -\lambda_{0}\ln p_{B_{0,i}*}-\lambda_{1}\ln p_{*B_{1,i}} \ge
\lambda_{0}\ln n_{0}+\lambda_{1}\ln n_{1} - \ln W_{i} \eeq
Clearly
\beq K(R_{*,\cdot\cdot}\|P_{\cdot\cdot})=-\ln S  \eeq
\beq K(R_{i,\cdot\cdot}\|P_{\cdot\cdot})=-\sum_{ijk}r_{i,jk}
\ln(r_{i,**}S)= -\ln S-\sum_{i}r_{i,**}\ln r_{i,**} \eeq
Now all the pieces come together:
\berrr I(P,\lambda_{0},\lambda_{1},\mu)\ge
\lambda_{0}\ln n_{0}+\lambda_{1}\ln n_{1}-
\sum_{i}r_{i,**}\ln W_{i}+
\mu\ln S+\sum_{i}r_{i,**}\ln r_{i,**}=\\=
\lambda_{0}\ln n_{0}+\lambda_{1}\ln n_{1}+\mu\ln S+
\sum_{i}r_{i,**}\ln\frac{r_{i,**}}{W_{i}} \eerrr
Another call of duty for lemma \ref{fun} produces
\beq \sum_{i}r_{i,**}\ln\frac{r_{i,**}}{W_{i}}\ge -\ln W \eeq
\end{proof}

\begin{lemma} \label{sublinear}
Suppose that
\beq (P_{1},d,n_{0,1},n_{1,1},S_{1},W_{1}),\ 
(P_{2},d,n_{0,2},n_{1,2},S_{2},W_{2}) \eeq
are attainable. Then
\beq (P_{1}\times P_{2},d,n_{0,1}n_{0,2},n_{1,1}n_{1,2},S_{1}S_{2},W_{1}W_{2})
\eeq
is attainable, where $\times$ is tensor product. In particular when
$ P_{1}=P_{2}=P $
for any $k_{1},k_{2}\ge 0$ we attain
\beq (P,(k_{1}+k_{2})d,n_{0,1}^{k_{1}}n_{0,2}^{k_{2}},
n_{1,1}^{k_{1}}n_{1,2}^{k_{2}},S_{1}^{k_{1}}S_{2}^{k_{2}},
W_{1}^{k_{1}}W_{2}^{k_{2}}) \eeq
In particular the closure of the log-attainable set $D^{c}(P)$ is convex.
\end{lemma}
\begin{proof}
Tensor product the codes.
\end{proof}

\begin{lemma} \label{sublinearplus}
Suppose that
\beq (P,d_{1},n_{0},n_{1},S_{1},W_{1}),\ (P,d_{2},n_{0},n_{1},S_{2},W_{2}) \eeq
are attainable. Then
\beq (P,d_{1}+d_{2},n_{0},n_{1},S_{1}+S_{2}-S_{1}S_{2},W_{1}+W_{2}) \eeq
is attainable. In particular for any $S_{1}\le\tilde{S}_{1}\le 1$
\beq (\ln n_{0},\ln n_{1},-\ln S_{1}/\tilde{S}_{1},\ln W_{1}/\tilde{S}_{1})
\in D_{0}^{c}(P) \eeq
\end{lemma}
\begin{proof}
Concatenating the codes shows the first claim.
Concatenating $T$ times the $k$'th tensor power of the first code shows that
\beq \left(P,Td_{1}^{k},n_{0}^{k},n_{1}^{k},1-\left(1-S_{1}^{k}\right)^{T},
TW_{1}^{k}\right) \eeq
is attainable. Taking
$ T=\lceil \tilde{S}_{1}^{-k}\rceil $
and letting $k\rightarrow\infty$
finishes the proof.
\end{proof}

Proof of {\bf theorem \ref{asy}}.
\begin{proof}
First let us show that the two representations are equivalent.
Denote the right hand side of (\ref{conv}) by $E$. It is the dual of its dual:
\berrr E=\{(m_{0},m_{1},s,w)\ |\ \alpha_{0}m_{0}+\alpha_{1}m_{1}-\beta s-
\gamma w\le 1\\ \forall \alpha_{0},\alpha_{1},\beta,\gamma,R\ \ 
{\rm such\ that}\ \ \alpha_{0},\alpha_{1}\le\gamma\le\alpha_{0}+\alpha_{1},
\ \beta,\gamma\ge 0,\\
\alpha_{0}\sum_{i}K(R_{i,\cdot *}\|P_{\cdot *})+
\alpha_{1}\sum_{i}K(R_{i,*\cdot}\|P_{*\cdot})+
(\gamma-\beta)K(R_{*,\cdot\cdot}\|P_{\cdot\cdot})-
\gamma\sum_{i}K(R_{i,\cdot\cdot}\|P_{\cdot\cdot})\le 1\} \eerrr
When $\gamma=0$ it forces $\alpha_{0}=\alpha_{1}=0$ and we are left with 
$-\beta s\le 1$ for all $\beta\ge 0$, i.e. $s\ge 0$.
When $\gamma>0$ we can divide by it, denote
$\lambda_{0}=\alpha_{0}/\gamma, \lambda_{1}=\alpha_{1}/\gamma,
\mu=\beta/\gamma$ and find that $1/\gamma\ge I$ so $E$ equals the right hand
side of (\ref{clos}).

Theorem \ref{big} implies that $D^{c}(P)\subset E$. We will prove the inverse
inclusion by construction. The single big bags pair code
\beq B_{0}=\{0,1,\ldots,b_{0}-1\},\ B_{1}=\{0,1,\ldots,b_{1}-1\} \eeq
shows that $\ D(0)\subset D(P) $ . Now let
$\{r_{i,jk}\}_{ijk}$ attain the bucketing information value $I$.
For dimension $d$ choose integers $\{d_{i,jk}\}_{ijk}$ such that
$ d_{*,**}=d $ and
\beq r_{i,jk}d-1<d_{i,jk}<r_{i,jk}d+1 \eeq

Let us define a bucket pair
\beq B_{0,0}=\left\{x_{0}\ \left|\ \forall ij\ \sum_{l=c_{i}+1}^{c_{i+1}}
(x_{0,l}==j)\ =\ d_{i,j*} \right.\right\} \eeq
\beq B_{0,1}=\left\{x_{1}\ \left|\ \forall ik\ \sum_{l=c_{i}+1}^{c_{i+1}}
(x_{1,l}==k)\ =\ d_{i,*k} \right.\right\} \eeq
where
$ c_{i}=\sum_{l=0}^{i-1}d_{i,**} $
In words we want $x_{0}$ to contain exactly $d_{0,j*}$ $j$-values in its first
$d_{0,**}$ coordinates, etc. The bucket size is
\beq p_{B_{0,0}*}=\prod_{i}\left[\frac{d_{i,**}!}{\prod_{j}d_{i,j*}!}
\prod_{j}p_{j*}^{d_{i,j*}}\right] \eeq
\beq p_{*B_{0,1}}=\prod_{i}\left[\frac{d_{i,**}!}{\prod_{k}d_{i,*k}!}
\prod_{k}p_{*k}^{d_{i,*k}}\right] \eeq
Let us add $T-1$ similar buckets. They are generated by randomly permuting
the coordinates $1,2,\ldots,d$. Let
$ n_{0}=1/p_{B_{0,0}*} $ ,
$ n_{1}=1/p_{*B_{0,1}} $
so that the work is
$ W=T $ .
A lower bound of the average success probability of this random bucketing code
is
\beq  {\rm E}[S]\ge U\left[1-(1-V/U)^{T}\right] \eeq
where
\beq U=\frac{d!}{\prod_{jk}d_{*,jk}!}\prod_{jk}p_{jk}^{d_{*,jk}} \eeq
is the probability that the special pair obtains coordinate pair $(j,k)$
exactly $d_{*,jk}$ times, and
\beq V=\prod_{i}\left[\frac{d_{i,**}!}{\prod_{jk}d_{i,jk}!}
\prod_{jk}p_{jk}^{d_{i,jk}}\right] \eeq
is the probability that the special pair obtains coordinate pair $(j,k)$
exactly $d_{i,jk}$ times in coordinate subset number $i$.
Of course there exists a deterministic code at least as successful as the
average code.

It is reasonable to take
$ T=\lceil U/V\rceil $ .
Stirling's approximation implies
\beq \lim_{d\rightarrow\infty}\frac{1}{d}\ln n_{0}=
\sum_{ij}r_{i,j*}\ln\frac{r_{i,j*}}{r_{i,**}p_{j*}} \eeq
\beq \lim_{d\rightarrow\infty}\frac{1}{d}\ln n_{1}=
\sum_{ik}r_{i,*k}\ln\frac{r_{i,*k}}{r_{i,**}p_{*k}} \eeq
\beq \lim_{d\rightarrow\infty}\frac{-1}{d}\ln U=
\sum_{jk}r_{*,jk}\ln\frac{r_{*,jk}}{p_{jk}} \eeq
\beq \lim_{d\rightarrow\infty}\frac{-1}{d}\ln V=
\sum_{ijk}r_{i,jk}\ln\frac{r_{i,jk}}{r_{i,**}p_{jk}} \eeq
Hence
\berr  \liminf_{d\rightarrow\infty}\frac{1}{d}(\lambda_{0}\ln n_{0}+
\lambda_{1}\ln n_{1}+\mu\ln S-\ln W) \ge \\ \ge
\lim_{d\rightarrow\infty}\frac{1}{d}(\lambda_{0}\ln n_{0}+
\lambda_{1}\ln n_{1}+(\mu-1)\ln U+\ln V)=I \eerr

There remains the unlimited data formula (\ref{unli}).
Lemmas \ref{sublinearplus} shows that
\beq D_{0}^{c}(P)=\tilde{D}_{0}(P)\cap\{(m_{0},m_{1},s,w\ |\ s\ge 0\} \eeq
\beq \tilde{D}_{0}(P)=D_{0}^{c}(0)+{\rm Cone}(\{(0,0,-1,1)\}) \eeq
Clearly $\tilde{D}_{0}(P)$ is convex, contains the origin, and any point
$(\alpha_{0},\alpha_{1},\beta,\gamma)$ in its dual satisfies
$ \beta\le\gamma $ . Hence $\mu=\beta/\gamma\le 1$
so by lemma \ref{equivalent} only one $i$ term is needed, as long as we use
the full $D_{0}(0)$.
\end{proof}

\end{spacing}

\end{document}